\documentclass[journal, 11pt, draftclsnofoot, onecolumn]{IEEEtran}
\usepackage{amsmath,graphicx}
\usepackage{amssymb}%
\usepackage{algorithmic}
\usepackage{cite}
\usepackage{multirow}
\usepackage{setspace}
\usepackage{color}

\DeclareMathOperator*{\argmax}{arg\,max}

\newtheorem{theorem}{Theorem}

\newtheorem{assumption}{Assumption}

\newtheorem{example}{Example}

\ifodd 0
\else

\fi

\title{\lowercase{e}Tutor: Online Learning for Personalized Education}
\author{\IEEEauthorblockN{Cem Tekin, Mihaela van der Schaar\\}
\IEEEauthorblockA{Electrical Engineering Department,
University of California, Los Angeles\\
Email: cmtkn@ucla.edu, mihaela@ee.ucla.edu}
}

\begin{document}

\maketitle

\begin{abstract}
Given recent advances in information technology and artificial
intelligence, web-based education systems have became complementary and, in some cases, viable alternatives to traditional
classroom teaching.
The popularity of these systems stems from their ability to make education available to a large demographics (see MOOCs). However, existing systems do not take advantage of the personalization which becomes possible when web-based education is offered: they continue to be one-size-fits-all. In this paper, we aim to provide a first systematic method for designing a personalized web-based education system. Personalizing education is challenging:  ($i$)  students need to be provided personalized teaching and training depending on their contexts (e.g. classes already taken, methods of learning preferred, etc.),
($ii$) for each specific context, the best teaching and training method (e.g type and order
of teaching materials to be shown) must be learned, ($iii$) teaching and training 
should be adapted online, based on the scores/feedback (e.g. tests, quizzes, final exam, likes/dislikes etc.) of the students. Our personalized online system, e-Tutor, is able to address these challenges by learning how to adapt the teaching methodology (in this case what sequence of
teaching material to present to a student) to maximize her performance
in the final exam, while minimizing the time spent by the students to learn the course (and possibly
dropouts). We illustrate the efficiency of
the proposed method on a real-world eTutor platform which
is used for remedial training for a Digital Signal Processing (DSP) course.
\end{abstract}

\begin{keywords} Online learning, personalized education, eLearning,
intelligent tutoring systems. \end{keywords}

\section{Introduction}\label{sec:intro} 
 
The last decade has witnessed an explosion in the number of web-based
education systems due to the increasing demand in higher-level education \cite{Survey}, limited number of teaching personnel, and advances
in information technology and artificial intelligence. Nowadays, most
universities have integrated \textit{Massive Open Online Course} (MOOC)
platforms into their education systems such as edX consortium, Coursera
or Udacity \cite{Stanford,Coursera,edX}, to give students the possibility to learn by interacting
with a software program instead of human teachers. Several advantages
of these systems over traditional classroom teaching are: (\emph{i})
they provide flexibility to the student in choosing what to learn
and when to learn, (\emph{ii}) they do not require the presence of
an interactive human teacher, (\emph{iii}) there are no limitations
in terms of the number of students who can take the course. However,
there are significant limitations of currently available online teaching
platforms. Since courses are taken online, there is no interaction
between the students and the teacher as in a classroom setting. This
makes it very difficult to meet the personalized needs of each student,
which may arise due to the differences between qualifications, learning
methods and cognitive skills of the students. It is observed that
if the personalization of teaching content is not carried out efficiently,
high drop-outs will occur \cite{Survey}. For instance, the students
that are very familiar with the topic may drop-out if the teaching
material is not challenging enough, while the students that are new to
the topic may get overstrained if the teaching material is hard. 

Due to these challenges, a new web-based education system that personalizes
education by learning online the needs of the students based on their
contexts, and adapting the teaching material based on the feedback
signals received from the student (answers to questions, quizzes,
etc.) is required. For this purpose we develop the \emph{eTutor} (illustrated in Fig. \ref{fig:basicfigure1}),
which is an online web-based education system, that learns how to teach a course, a concept or remedial materials to a student with a specific context in the most efficient
way. Basically, for the current student, eTutor learns from its past
interactions with students with similar contexts, the sequence of
teaching materials that are shown to these students, and the response
of these students to the teaching materials including the final exam
scores, how to teach the course in the most effective way. This is
done by defining a \emph{teaching effectiveness metric}, referred to as the {\em regret}, that is a function of the
final exam score and time cost of teaching to the student, and then
designing a learning algorithm that learns to optimize this metric.
This tradeoff between learning (exploring) and optimizing (exploiting) is captured by the eTutor in the most efficient way,
i.e., the average exam score of the students converge to the average
exam score that could be achieved by the best teaching strategy. We illustrate the efficiency
of the proposed system in a real-world experiment carried out on students
in a DSP class. 

\begin{figure}
\begin{center}
\includegraphics[scale=0.23]{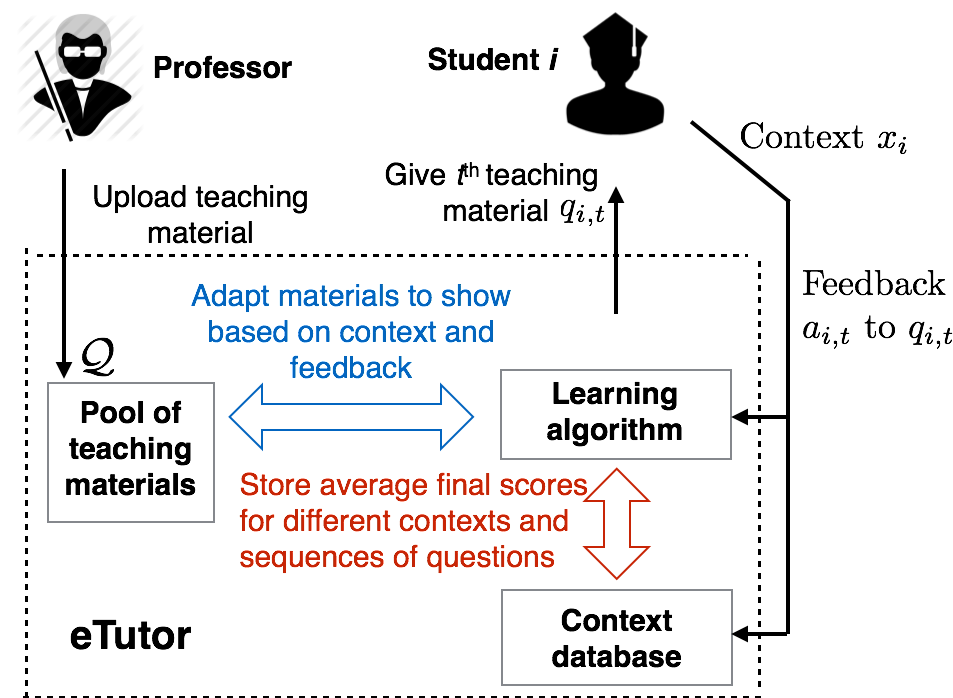}
\protect\caption{eTutor, student and professor interaction. \label{fig:basicfigure1}}
\end{center}
\end{figure}

\subsection{Related Work}  
Although web-based education systems have recently become popular,
there is no consensus or standards on how to design an optimal web-based
education system. A detailed comparison of our work with the related
work in web-based education is given in Table \ref{table:relwork}. Most of the recent
works focus on the subfield of MOOCs, which are online courses with
very large number of students \cite{Survey, Stanford, Coursera, edX}. Among these, several works examine
students' interaction with commercially available MOOC systems such
as Coursera \cite{Stanford,Coursera} and edX \cite{edX}.

Apart from these, two approaches exist in designing web-based education
systems: adaptive education systems and intelligent tutoring systems.
In an adaptive education system \cite{AES1,AES2}, the teaching materials
that are shown to each student are adapted based on the context of
the student, but not based on the feedback the student provides during
the course. This adaptation is based on numerous contexts including
the student's learning style, her knowledge, background, origin, grades,
previously taken courses etc. In contrast, in an intelligent tutoring
system adaptation is done based on the response of the student to
the given teaching material \cite{IT1,IT2,IT3,IT4,IT5,IT6,IT7}, without
taking into account contexts. Our work combines both ideas by adapting
the sequence of teaching materials that is presented to a student
based on both the context and the feedback of the student. However,
our techniques are very different from both lines of research. Our
goal is to learn the optimal way to teach a course {\em in a way that is most effective for each student}. To learn effectively, our method utilizes the past knowledge gained about the efficacy of the material from students with similar contexts who have taken the course before. This is different
from \cite{IT1,IT2,IT3,IT4,IT5,IT6,IT7}, which only take into account
the current student's response to the previously shown teaching material.


\begin{table}
\centering
{\fontsize{10}{10}\selectfont
\setlength{\tabcolsep}{.1em}
\begin{tabular}{|l|c|c|c|c|}
\hline
Method & Context-based  & Feedback-based &  Learns from & Regret \\
       & learning		  & learning    &   final exam & bound \\
\hline
\cite{AES1,AES2} & Yes & No & No & No \\
\hline
\cite{IT1,IT2,IT3,IT4,IT5,IT6,IT7} & No & Yes & No & No \\
\hline
Our work & Yes & Yes & Yes & Yes \\
\hline
\end{tabular}
}
\caption{Comparison with related work.}
\label{table:relwork}
\vspace{-0.3in}
\end{table} 
\section{Formalism, Algorithm and Analysis} \label{sec:Formalism_Algorithm_Analysis} 
In this section we mathematically formalize the online teaching/tutoring problem, define a benchmark tutor (i.e. the "ideal" tutor) and propose an online learning algorithm for the eTutor which converges in performance to the benchmark tutor that knows the optimal sequence of teaching materials to show for each student.

\subsection{Problem Definition}

Consider a set of students participating in an online education system
and a concept that should be learned by the students. The comprehension
of the concept will be tested via a final exam (test). We assume that
the students arrive sequentially over time and use index $i$ to denote the $i$th student. Additionally, we assume that when a student first
interacts with the online education system, she needs to answer a
set of questions, which will form the $\emph{context}$ of the student.
Context may include information about the student such as age, grades, whether she prefers visual or written instructions,
etc. Denote the finite set of all possible contexts by ${\cal X}$ and an
element of ${\cal X}$ by $x$. The concept will be taught by presenting
a set of teaching materials (written or visual) to the student and
asking a set of questions about these materials and providing their
answers. Let ${\cal Q}$ be the set of teaching materials (consists
of text/images to learn from and questions) that can be given to the
student. The number of elements of ${\cal Q}$ is denoted by $Q$. 

The materials that are shown to a student are chosen in an online
way based on the context of the student, previous materials that are
shown to the student, the student's response to shown questions (whether
the answer is correct or not) and all the previous knowledge obtained
from past students with contexts, responses and scores similar to
the current student. It is also important to learn in which order the
materials should be shown, since learning from one material may require
knowledge of a concept which can be learned by understanding another
material. 

For each student $i$, we consider a discrete time model $t=1,2,\ldots,T_{i}$,
where time $t$ denotes the sequence of events related to the $t$th
material that is shown to the student. $T_{i}$ denotes the number
of teaching materials shown to student $i$ before the final exam
is given (depends on student's feedback). Clearly, $T_{i}\leq Q$. The $t$th teaching
material shown to student $i$ is denoted by $q_{i,t}$. Let $\boldsymbol{q}_{i}:=(q_{i,1},\ldots,q_{i,T_{i}})$,
and $\boldsymbol{q}_{i}[t] := (q_{i,1},\ldots,q_{i,t})$.

We denote student $i$'s response to $q_{i,t}$ by $a_{i,t}\in {\cal A}$, where ${\cal A}$ is the set of possible feedbacks that the student can provide to a teaching material. We assume that ${\cal A}$ is finite. For instance, an example can be the case when ${\cal A} := \{ -1, 0, 1  \}$. 
If the student does not
provide any feedback on the teaching material, we have $a_{i,t}=0$; when the teaching
material is a (multiple-choice) question, $a_{i,t}=1$ denotes a correct
answer and $a_{i,t}=-1$ denotes a wrong answer. 
Let $\boldsymbol{a}_i := (a_{i,1},\ldots,a_{i,T_i})$ and 
$\boldsymbol{a}_i[t] := (a_{i,1},\ldots,a_{i,t})$, $t \leq T_i$. In addition, let $a_{i,0} :=0$, which indicates that no feedback is available prior to $1$st teaching material.

Let $\boldsymbol{S}$ denote the set of all sequences of teaching materials
that can be shown.\footnote{In practice, it is possible to give $\boldsymbol{S}$ as an input in addition to ${\cal Q}$. For instance, some sequences which are classified by the professor as unreasonable can be discarded, significantly reducing the size of $\boldsymbol{S}$.} 
We have 
\begin{align}
|  \boldsymbol{S}  | = \sum_{t=1}^Q {Q \choose t} t! = Q! \sum_{t=1}^Q \frac{1}{(Q-t)!} 
= Q! \left(  \frac{e \Gamma(Q+1,1)}{\Gamma(Q+1)} - 1 \right) \geq Q! ,     \notag
\end{align}
where $\Gamma(y)$ is the gamma function and $\Gamma(x,y)$ is the incomplete gamma function. 

For a sequence of materials $\boldsymbol{s}\in\boldsymbol{S}$,
let $\boldsymbol{{\cal A}}(\boldsymbol{s})$ be the set of sequences
of feedbacks a student can provide. The expected final exam score
for a student with context $x$, sequence of questions $\boldsymbol{s}\in\boldsymbol{S}$
and sequence of feedbacks $\boldsymbol{a}\in\boldsymbol{{\cal A}}(\boldsymbol{s})$
is denoted by $r_{x,\boldsymbol{s},\boldsymbol{a}}$. We assume that
the final exam score of a student with context $x$, the sequence
of teaching materials $\boldsymbol{s}$ and the sequence of feedbacks
$\boldsymbol{a}$ is randomly drawn from a $F_{x,\boldsymbol{s},\boldsymbol{a}}$
with expected value $r_{x,\boldsymbol{s},\boldsymbol{a}}$. Both $F_{x,\boldsymbol{s},\boldsymbol{a}}$
and $r_{x,\boldsymbol{s},\boldsymbol{a}}$ are unknown. 

\subsection{The Benchmark Tutor}

Due to the enormous number of possible sequences of teaching materials, it is
not possible to learn the best sequence of teaching materials by trying
all of them for different students. In this section we define a benchmark tutor, whose teaching strategy can be learned very fast. We call it the \emph{best-first} (BF) benchmark. Due to limited space its pseudocode is given in Fig. \ref{fig:BS}, however, we describe it in detail below. In order to explain this benchmark, we require a few more notations.

\begin{figure}[htb]
\fbox {\begin{minipage}{0.95\columnwidth}
\begin{algorithmic}[1]
\STATE{Receive student context $x$}
\STATE{Show $q^*_{x,1} = \argmax_{ q \in {\cal Q} } y_{x,q,0}$}
\STATE{Receive $a^*_1$.}
\WHILE{$1<t\leq Q$}
\IF{\flushleft{$r_{x,\boldsymbol{s}^*[t-1],\boldsymbol{a}^*[t-1]} 
\geq \max_{q \in {\cal Q}_{\boldsymbol{s}^*[t-1]} } y_{x,\boldsymbol{s}^*[t-1], \boldsymbol{a}^*[t-1] }-c$}}
\STATE{Give the final exam. //BREAK}
\ELSE
\STATE{$q^*_{x,t} = \argmax_{q \in {\cal Q}_{\boldsymbol{s}^*[t-1]}}  y_{x,\boldsymbol{s}^*[t-1], \boldsymbol{a}^*[t-1]}$}
\ENDIF
\STATE{$t=t+1$}
\ENDWHILE
\end{algorithmic}
\end{minipage}
}  
\vspace{-0.1in}
\caption{Pseudocode for BF.}
\label{fig:BS}
\vspace{-0.1in}
\end{figure}

Given a sequence $\boldsymbol{s}$ of teaching materials, let ${\cal Q}_{\boldsymbol{s}}$
be the set of remaining teaching materials that can be given to the
student. Let $\boldsymbol{S}[t]\subset\boldsymbol{S}$ be the set
of sequences that consists of $t$ teaching materials followed by
the final exam. In order to explicitly state the number of teaching
materials in a sequence of teaching materials, we will use the notation
$\boldsymbol{s}[t]$ to denote an element of $\boldsymbol{S}[t]$.
We will also use $\boldsymbol{a}_{\boldsymbol{s}[t]}[t']$ to denote
the student's feedback to the first $t'$ teaching materials in $\boldsymbol{s}[t]$.
Let 
\begin{align}
y_{x,\boldsymbol{s}[t],\boldsymbol{a}_{\boldsymbol{s}[t]}[t-1]}
:= \mathrm{E}_{a_{t}}[r_{x,\boldsymbol{s}[t],(\boldsymbol{a}_{\boldsymbol{s}[t]}[t-1],a_{t})}] ,      \notag
\end{align}
be the \emph{ex-ante} final exam score of a student with context $x$
which is given teaching materials $\boldsymbol{s}[t]$ and provided
feedback to all of them except the last teaching material. 

The BF benchmark incrementally selects the next teaching material
to show based on the student's feedback about the previous teaching
materials. The first teaching material it shows is $q^*_{x,1} = \argmax_{ q \in {\cal Q} } y_{x,q,0}$.
Let $\boldsymbol{q}^*_{x} = (q^*_{x,1}, q^*_{x,2}, \ldots, q^*_{x,T} )$ be the sequence of teaching materials shown by the BF benchmark to a student with context $x$, where $T$ is the random total number of materials shown to the student, which depends on the feedback of the student to the shown materials.
In general, the $t$th teaching material to show, i.e., $q^*_{x,t}$, depends on both $\boldsymbol{q}^*_x[t-1]$ and $a_{\boldsymbol{q}^*_x[t-1]}[t-1]$. We assume that the following property holds for the expected final exam score given the sequence of materials shown by the BF benchmark and the feedback obtained from these shown materials:       

\begin{assumption}\label{ass:Markov}
Consider a context $x \in {\cal X}$ and any two sequences $(\boldsymbol{q}^*, \boldsymbol{a})$ and $(\boldsymbol{q}^*, \boldsymbol{a}')$, where $\boldsymbol{q}^*$ is the set of teaching materials shown by the BF benchmark and $\boldsymbol{a}$ and $\boldsymbol{a}'$ are two feedback sequences that are associated with this set of teaching materials. Then, we have
\begin{align}
\argmax_{ q \in {\cal Q}_{\boldsymbol{q}^*[t]}  } y_{x, (\boldsymbol{q}^*[t], q) , \boldsymbol{a}[t]}    
=  \argmax_{ q \in {\cal Q}_{\boldsymbol{q}^*[t]}  } y_{x, (\boldsymbol{q}^*[t], q) , \boldsymbol{a}'[t]}  .
 \notag
\end{align} 
\end{assumption}

We have $q^*_{x,t}=\argmax_{q \in {\cal Q}_{\boldsymbol{q}^*_x[t-1]} } r_{x,\boldsymbol{q}^*_x[t-1], \boldsymbol{a}^*_{t-1}}$. 
For any $t$, if $r_{x,\boldsymbol{q}^*_x[t],\boldsymbol{a}^*[t]}\geq y_{x,(\boldsymbol{q}^*_x[t],q),\boldsymbol{a}^*[t]}-c$
for all $q\in{\cal Q}(\boldsymbol{q}^*_x[t])$, then the BF benchmark
will give the final exam after the $t$th teaching material. Here $c>0$ is the {\em teaching cost} of showing one more material to the student, which is the cost related to the time it takes for the student to complete the teaching material. 
The average final exam score minus the teaching cost achieved by following the BF benchmark for
the first $n$ students is equal to 
\begin{align}
RW_{\textrm{BF}}(n) = \sum_{i=1}^{n}  
\frac{  \mathrm{E} [Y_{x_{i},\boldsymbol{Q}_{i,}^{*}, \boldsymbol{A}_{i}^{*}} 
- c |\boldsymbol{Q}_{i,}^{*}|]} {n},   \notag
\end{align}
where $Y_{x_{i},\boldsymbol{Q}_{i,}^{*}\boldsymbol{A}_{i}^{*}}$ is
the random variable that represents the final exam score of student
$i$, where $\boldsymbol{Q}_{i}^{*}$ is the random variable that
represents the sequence of teaching materials given to student $i$
by the BF benchmark, and $\boldsymbol{A}_{i}^{*}$ is the
random variable that represents the sequence of feedbacks provided
by student $i$ to the teaching materials $\boldsymbol{Q}_{i}^{*}$. The BF benchmark is an {\em oracle} policy because we assume
that nothing is known about the expected exam scores a priori. Any
learning algorithm $\alpha$ which selects a sequence of teaching materials
$\boldsymbol{Q}^{\alpha}_{i}$ based on the sequence of feedbacks $\boldsymbol{A}^\alpha_{i}$
has a average regret with respect to the BF benchmark which
is given by 
\begin{align}
R(n)=RW_{BF}(n)-   \sum_{i=1}^{n}  
 \frac{ \mathrm{E} [Y_{x_{i},\boldsymbol{Q}_{i,}^{\alpha}\boldsymbol{A}_{i}^{\alpha}}  - |\boldsymbol{Q}^{\alpha}_{i}| ]} {n} .    \label{eqn:regretdef}
\end{align}


\begin{figure}
\begin{center}
\includegraphics[scale=0.2]{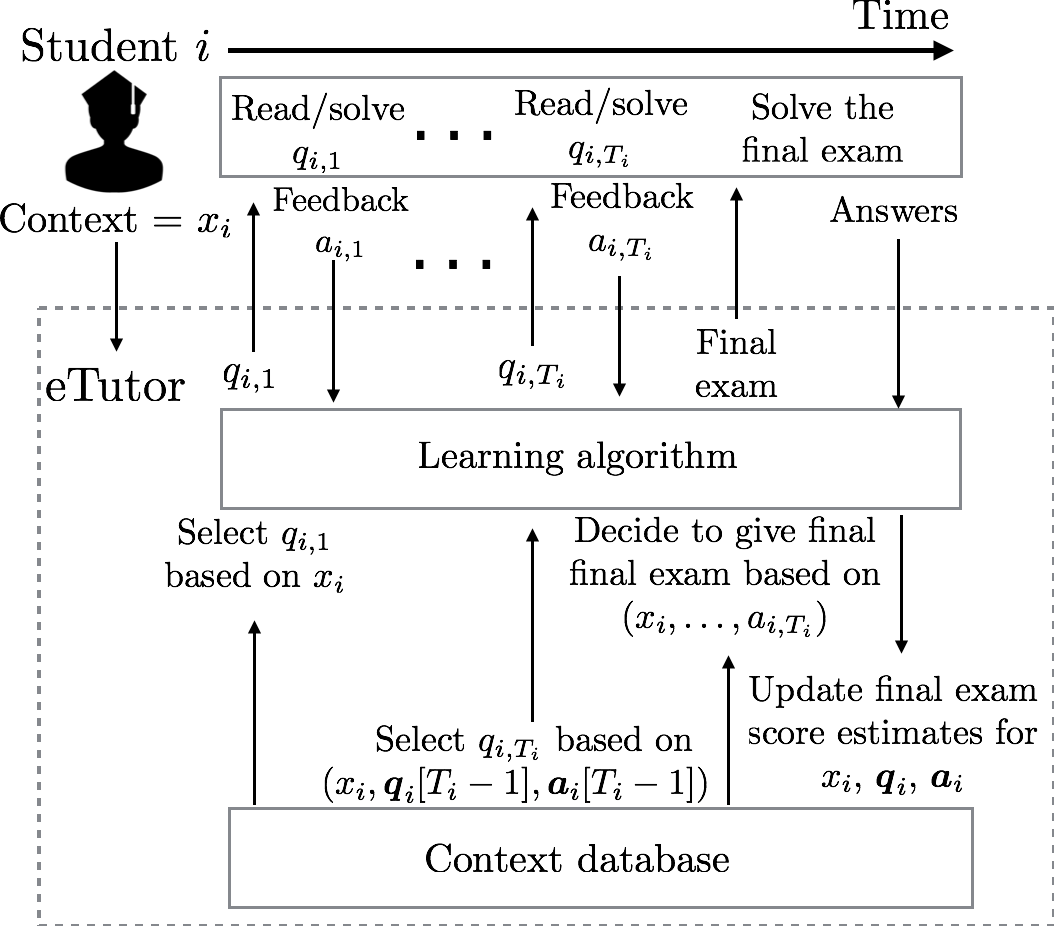}
\vspace{-0.2in}
\protect\caption{Operation of the eTutor. \label{fig:Basic2}}
\vspace{-0.4in}
\end{center}
\end{figure}

The next example illustrates that when the cost of showing new materials, i.e., $c$, is high, BF benchmark is better than the best fixed sequence. 
\begin{example}
Consider ${\cal Q} = \{ a,b \}$, ${\cal A} = \{ 0, 1\}$ and ${\cal X} = \{ x \}$. Assume that the expected rewards are given as follows: $r_{x,a,0} =0$, $r_{x,b,0} =0$, $r_{x,a,1} = 12$, $r_{x,b,1} = 6$, 
$r_{x, (a,b), (1,1)} = 13$, $r_{x, (a,b), (1,0)} =12$, $r_{x, (a,b), (0,1)} = 10$, $r_{x,(a,b),(0,0)}=9$. 
Let $\mathrm{P}_x(\boldsymbol{a} | \boldsymbol{q})$ denote the probability that feedback sequence $\boldsymbol{a}$ is given by a student with context $x$ as a response to the sequence $\boldsymbol{q}$. Assume that we have
$\mathrm{P}_x( 1 | a  ) = 0.5$, $\mathrm{P}_x( 0 | a  ) = 0.5$, $\mathrm{P}_x(  (0,0) | (a,b)  ) = 0.3$, $\mathrm{P}_x(  (0,1) | (a,b)  ) = 0.2$, $\mathrm{P}_x(  (1,1) | (a,b)  ) = 0.4$, $\mathrm{P}_x(  (1,0) | (a,b)  ) = 0.1$.

The BF benchmark shows $a$ as the first teaching material. Then if feedback is $0$ it also shows $b$ before giving the final exam. Else, it gives the final exam just after showing $a$. Hence the expected reward of the BF is 
\begin{align}
RW_{\textrm{BF}}(n) = 0.5 \times 12 + 0.3 \times (9 -c ) + 0.2 \times (10 -c ) = 10.7 - 0.5c ,  \notag
\end{align}
where $c$ is the cost of showing the second material. 
The best fixed sequence always shows $a$ first, $b$ second, and then gives the final exam. Its expected reward is equal to 
\begin{align}
0.3 \times 9 + 0.4 \times 12 + 0.2 \times 10 + 0.1 \times 11 - c = 11 - c.     \notag
\end{align}
Clearly, the BF benchmark is better than the best fixed sequence given that $c > 3/5$.
\end{example}

\subsection{\lowercase{e}Tutor} \label{sec:percel} 
In this section we propose {\em eTutor} (pseudocode given in Fig. \ref{fig:eTutor}), which learns the optimal sequence of teaching materials to show based on the student's context and feedback about the previously shown teaching materials (as shown in Fig. \ref{fig:Basic2}). In order to minimize the regret given in (\ref{eqn:regretdef}), eTutor balances exploration and exploitation when selecting the teaching materials to show to the student. Consider a student $i$ and the $t$th teaching material shown to that student. eTutor keeps the following sample mean reward estimates:
($i$) $\hat{r}_{x,t,q,a}(i)$ which is the estimated final exam score for students with context $x$ that took the course before student $i$ who are given the final exam right after material $q$ is given as the $t$th material and feedback $a$ is observed, ($ii$) $\hat{y}_{x,a,t,q}(i)$ which is the estimated final exam score for students with context $x$ that took the course before student $i$ who are given the final exam right after material $q$ is given as the $t$th material after observing feedback $a$ for the $t-1$th material. In addition to these, eTutor keeps the following counters: ($i$) $T_{x,t,q,a}(i)$ which counts the number of times material $q$ is shown as the $t$th material and feedback $a$ is obtained for students with context $x$ that took the course before student $i$, ($ii$) $T_{x,a,t,q}(i)$ which counts the number of times material $q$ is shown as the $t$th material after feedback $a$ is obtained from the previously shown material for students with context $x$ that took the course before student $i$. 

\begin{figure}[h!]
\fbox {\begin{minipage}{0.95\columnwidth}
\begin{algorithmic}[1]
\STATE{Input $D>0$, $\delta>0$.}
\STATE{Initialize: $\hat{r}_{x,t,q,a} =0, \hat{y}_{x,a,t,q}=0$, $T_{x,t,q,a}=0$, $T_{x,a,t,q}=0$, $\forall x \in {\cal X}, a \in \{ -1,0,1\}, q \in {\cal Q}, t=1,\ldots,Q$. $a_{i,0}=0$, $\boldsymbol{q}_i[0] = \emptyset$, $\forall i =1,2,\ldots$.}
\WHILE{$i \geq 1$}
\STATE{Receive student context $x = x_i$}
\STATE{${\cal U}_1 = \{ q \in {\cal Q} : T_{x,0,1,q} < D \log (i / \delta) \}$  }
\IF{${\cal U}_1 \neq \emptyset$}
\STATE{Give $q_{i,1}$ randomly selected from ${\cal U}_1$, get $a_{i,1}$.}
\STATE{Give the final exam, get the score $X(i)$, $t^* = 1$, //BREAK}
\ELSE
\STATE{Give $q_{i,1} = \argmax_{q \in {\cal Q}}  \hat{y}_{x,0,1,q}$, get $a_{i,1}$.}
\ENDIF
\STATE{$t=2$}
\WHILE{$2 \leq t \leq Q$}
\STATE{${\cal U}_t = \{ q \in {\cal Q}_{\boldsymbol{q}_i[t-1]} : T_{x,a_{i,t-1},t,q} < D \log  (i / \delta)   \}$}
\IF{$T_{x,t-1,q_{i,t-1},a_{i,t-1}} < D \log i$}
\STATE{Give the final exam, get the score $X(i)$, $t^* = t-1$, //BREAK}
\ELSIF{${\cal U}_t \neq \emptyset$}
\STATE{Show $q_{i,t}$ randomly selected from ${\cal U}_t$ and get the feedback $a_{i,t}$.}
\STATE{Give the final exam, get the score $X(i)$, $t^* = t-1$, //BREAK}
\ELSE
\IF{$\hat{r}_{x,t-1,q,a_{i,t-1}} \geq \hat{y}_{x,a_{i,t-1},t,q'} -c$, $\forall q \in {\cal Q}_{\boldsymbol{q}_i[t-1]}$}
\STATE{Give the final exam, get the score $X(i)$, $t^* = t-1$, //BREAK}
\ELSE
\STATE{Show $q_{i,t} = \argmax_{q \in {\cal Q}_{\boldsymbol{q}_i[t-1]}} \hat{y}_{x,a_{i,t-1},t,q'}$ and get the feedback $a_{i,1}$}
\ENDIF
\ENDIF
\STATE{$t=t+1$}
\ENDWHILE
\STATE{Update $\hat{r}_{x,t^*,q_{i,t^*},a_{i,t^*}}$, $\hat{y}_{x,a_{i,t^*-1},t^*,q_{i,t*}}$ using $X(i)$ (sample mean update).}
\STATE{$T_{ x,t^*,q_{i,t^*},a_{i,t^*} }++$, $T_{x,a_{i,t^*-1},t^*,q_{i,t^*}}++$.}
\STATE{$i=i+1$}
\ENDWHILE
\end{algorithmic}
\end{minipage}
}  
\vspace{-0.1in}
\caption{Pseudocode for eTutor.}
\label{fig:eTutor}
\vspace{-0.25in}
\end{figure}

Next, we explain how exploration and exploitation is performed. 
Consider the event that eTutor asks question $q_{i,t}=q$ and receives feedback $a_{i,t}=a$. It first checks if $T_{x,t,q,a}(i) < D \log  (i / \delta) $, where $D>0$ and $\delta>0$ are constants that are input parameters of eTutor. If this holds, then eTutor explores by giving the final exam and obtaining the final score $X(i)$, by which it updates $\hat{r}_{x,t,q,a}(i+1) = (\hat{r}_{x,t,q,a}(i+1) + X(i))/( T_{x,t,q,a}(i) + 1 )$. Else if $T_{x,t,q,a}(i) \geq D \log  (i / \delta) $, eTutor checks if there are any questions $q' \in {\cal Q}_{\boldsymbol{q}_i[t]}$ for which $T_{x,a_{i,t},t+1,q'}(i) < D \log  (i / \delta)$. If there are such questions, then eTutor explores one of them randomly by showing that material to the student, obtaining the feedback, giving the final exam, and obtaining the final exam score. The obtained final exam score $X(i)$ is used for both updating $\hat{r}_{x,t+1,q',a_{i,t+1}}(i+1)$ and $\hat{y}_{x,a_{i,t},t+1,q'}(i+1)$. If none of the above events happen, then eTutor exploits at $t$. To do this it first checks if 
$\hat{r}_{x,t,q,a_{i,t}}(i)   \geq \hat{y}_{x,a_{i,t},t+1,q'}(i) -c$, 
for all $q' \in {\cal Q}_{\boldsymbol{q}_i[t]}$. If this is the case, it means that showing one more teaching material does not increase the final exam score enough to compensate for the {\em teaching cost} of showing one more material. Hence, eTutor gives the final exam after its $t$th material. If the opposite happens, then it means that showing one more material can improve final exam score sufficiently enough for it to compensate the cost of teaching. Hence, eTutor will show one more teaching material to the student which is 
$q_{i,t+1} = \argmax_{q' \in {\cal Q}_{\boldsymbol{q}_i[t]}} \hat{y}_{x,a_{i,t},t+1,q'}(i)$. 
The next decision to take will be based on the student's feedback to $q_{i,t+1}$ which is $a_{i,t+1}$.
This goes on until eTutor gives the final exam, which will eventually happen since ${\cal Q}$ is finite.   
 
\subsection{Regret and Confidence Bounds For eTutor} \label{sec:regret}
 The regret of eTutor can be written as the sum of two separate regret terms: regret for the experimental materials shown to students when eTutor explores, i.e., $R_e(n)$, and regret for students that eTutor exploits, i.e., $R_s(n)$. Hence we can write $R(n) = E[R_e(n)] + E[R_s(n)]$. 

For a sequence of numbers $\{ r \}_{r \in {\cal R}}$, let $\mathrm{min2}(  \{ r \}_{r \in {\cal R}})$ be the difference between the highest and the second highest numbers. 
Consider any sequence of materials $\boldsymbol{q}^*_x[t] \in \boldsymbol{S}[t]$ and feedback $\boldsymbol{a}[t] \in \boldsymbol{{\cal A}}(\boldsymbol{q}^*_x[t])$, where $\boldsymbol{q}^*_x[t]$ is the sequence of materials shown by the BF benchmark to a student with context $x$. Let 
\begin{align}
\Delta_{\min,1} := \mathrm{min2}(  \{ y_{x, q, 0}   \}_{q \in {\cal Q}}     )     ,  \notag
\end{align}
and
\begin{align}
\Delta_{\min,t} := \mathrm{min2}(  r_{x,\boldsymbol{q}^*_x[t],\boldsymbol{a}[t]} ,    
 \{ y_{x, (\boldsymbol{q}^*_x[t], q), \boldsymbol{a}[t]}   \}_{q \in {\cal Q}_{\boldsymbol{q}^*_x[t]}}     )  ,     \notag
\end{align}
for $1< t < Q$. 
Let $\Delta_{\min} := \min_{t=1,\ldots, Q-1} \Delta_{\min,t}$.
Given that the constant $D$ that is input to eTutor is such that 
$D \geq 4/\Delta^2_{\min}$, where $\Delta_{\min} =  (\min_{ x \in {\cal X}, \boldsymbol{s} \in \boldsymbol{{\cal S}}, \boldsymbol{a} \in {\boldsymbol{\cal A}}(\boldsymbol{s})}     r_{x, \boldsymbol{s}, \boldsymbol{a}})^2$,
Assuming that the maximum final exam score is equal to $1$, we have the following bounds on the regret. 

\begin{theorem} \label{thm:regret}
Setting the parameters of eTutor as $D = 4/\Delta^2_{\min}$ and $\delta = \sqrt{\epsilon}/(Q \sqrt{2 \beta})$, where $\beta = \sum_{t=1}^\infty 1/t^2$, we have the following bounds on the regret of eTutor. 
The regret of eTutor for the first $n$ students is bounded as follows:\\
($i$) $R_{e}(n) \leq 2 |{\cal X}| |{\cal A}| Q^2 D\log ( n/ \delta)/n$ with probability 1. \\  
($ii$) $R_{s}(n)  = 0$ with probability at least $1-\epsilon$. \\
($iii$) $R(n)  \leq 2 |{\cal X}| |{\cal A}| Q^2 D  \log (n/ \delta)/n + \epsilon$.
\end{theorem}

\begin{proof}
Our proof involves showing that when the eTutor estimates the final exam scores for the sequences of materials it gives to the students such that they are within $\Delta_{\min}/2$ of the true final exam score, then it will always show the same set of teaching materials as the BF benchmark does. 

To proceed, we define the following sets of students.
Let $E_1(n)$ be the set of students in $\{1,\ldots,n\}$ for which eTutor explores the teaching material in the first slot, i.e., $i \in \{1,\ldots,n\}$ for which $T_{x_i, 0,1,q}(i) < D \log(i/\delta)$ for some $q \in {\cal Q}$ such that after the material is shown, the final exam is given. 
Let $E_t(n)$, $1<t \leq Q$ be the set of students in $\{1,\ldots,n\}$ for which eTutor explores in the $t$th slot, i.e., the set of time slots for which eTutor exploited up to the $t-1$th slot and $\bar{r}_{x_i, t-1, q_{i,t-1}, a_{i,t-1}}(i) \leq D \log(i/\delta)$ or 
$\hat{y}_{x_i, a_{i,t-1}, t, q } \leq D \log(i/\delta)$ for some $q \in Q_{\boldsymbol{q}_i[t-1]}$ such that the final exam is given either after the $t-1$th slot or the $t$th slot depending on which teaching material is under-explored.
Let $\tau_1(n)$ be the set of students in $\{1,\ldots,n\}$ for which eTutor exploits for the first slot, i.e., $T_{x_i,0,1,q} \geq D \log(i/\delta)$ for all $q \in {\cal Q}$. 
Let $\tau_t(n)$ be the set of students in $\{1,\ldots,n\}$ for which eTutor exploits for the $t$th slot, i.e., $\bar{r}_{x_i, t-1, q_{i,t-1}, a_{i,t-1}}(i) \geq D \log(i/\delta)$ and 
$\hat{y}_{x_i, a_{i,t-1}, t,q} \geq D \log(i/\delta)$ for all $q \in {\cal Q}_{\boldsymbol{q}_i[t-1]}$ such that eTutor has not given the final exam before slot $t-1$. 
Let 
$Z_t(n) := \tau_t(n) - \tau_{t+1}(n) - E_{t+1}(n)$, $1 \leq t < Q$ 
denote the set of students in $\{1,\ldots,n\}$  for which eTutor stops and gives the final exam after the $t-1$th teaching material is shown at times when it exploits. Let $Z_Q(n) := \tau_Q(n)$. 
The set of all students for which eTutor explores until the $n$th student is equal to 
$E(n) := \bigcup_{t=1}^Q E_1(n)$, where $E_t(n) \cap E_{t'}(n) = \emptyset$ for $t \neq t$. 
The set of all students for which eTutor exploits until the $n$th student is equal to 
$Z(n) := \bigcup_{t=1}^Q Z_t(n)$, where $Z_t(n) \cap Z_{t'}(n) = \emptyset$. 
We also have 
$Z(n) := \tau_1(n)$, $\tau_{t}(n) = \tau_{t+1}(n) \cup E_{t+1}(n) \cup Z_t(n)$ for $1 \leq t < Q$. 

In the following we will bound $R_{s}(n)$.
Next, we define the events which correspond to the case that the estimated final exam scores for the sequences followed by the BF benchmark are within $\Delta_{\min}/2$ of the expected final exam scores. 
Let
\begin{align}
\textrm{Perf}_1(n) := \{  |\hat{y}_{x_i, 0, 1, q} - y_{x_i, q, 0 }    | < \Delta_{\min}/2, \forall q \in {\cal Q},  \forall i \in \tau_1(n)  \}    ,  \notag
\end{align}
 and 
 \begin{align}
 \textrm{Perf}_t(n) &:= \left\{   |\hat{r}_{x_i, t-1, q_{i,t-1}, a_{i,t-1}} - r_{x_i, \boldsymbol{q}^*_{x_i}[t-1], \boldsymbol{a}_{\boldsymbol{q}^*_{x_i}[t-1]}[t-1]}    | < \Delta_{\min}/2, \right. \notag \\
&\left. |\hat{y}_{x_i, a_{i,t-1}, t, q} - y_{x_i, (\boldsymbol{q}^*_{x_i}[t-1], q) , \boldsymbol{a}_{\boldsymbol{q}^*_{x_i}[t-1]}[t-1] }    | < \Delta_{\min}/2, \forall q \in {\cal Q}_{\boldsymbol{q}^*_{x_i}[t-1]},  \forall i \in \tau_t(n)  \right\}       \notag
 \end{align}
Let
\begin{align}
\textrm{Perf}(n)  = \bigcap_{t=1}^Q  \textrm{Perf}_t(n).  \notag
\end{align}
On event $\textrm{Perf}(n)$, eTutor selects teaching materials for the students in the same way as BF benchmark does. Hence the contribution to the regret given in (\ref{eqn:regretdef}) on event $\textrm{Perf}(n)$ is zero. 

Next, we lower bound the probability of event $\textrm{Perf}(n)$.
Using the chain rule we can write 
\begin{align}
\mathrm{P} (\textrm{Perf}(n) )  &=  \mathrm{P} (\textrm{Perf}_Q(n), \textrm{Perf}_{Q-1}(n), \ldots, \textrm{Perf}_1(n) )  \notag \\
& = \mathrm{P} (\textrm{Perf}_Q(n) |  \textrm{Perf}_{Q-1}(n), \ldots, \textrm{Perf}_1(n) )  \times
\mathrm{P} (\textrm{Perf}_{Q-1}(n) | \textrm{Perf}_{Q-2}(n), \ldots, \textrm{Perf}_1(n) )  \notag \\
& \times \ldots \times \mathrm{P}  (\textrm{Perf}_{2}(n) |  \textrm{Perf}_{1}(n)   ) 
\times \mathrm{P}( \textrm{Perf}_{1}(n)  ).
\end{align}

For an event $E$, let $E^c$ denote its complement. Note that we have 
\begin{align}
 \mathrm{P} ( \textrm{Perf}_{1}(n)^c  )  
 & \leq \sum_{i \in \tau_1(n)} \sum_{q \in {\cal Q}} \mathrm{P} (|\hat{y}_{x_i, 0, 1, q} - y_{x_i, q, 0 }    | < \Delta_{\min}/2)     \notag \\
 &\leq \sum_{i \in \tau_1(n)}  2 Q \exp(-2D \log(i/\delta) \Delta^2_{\min}/4 ) \\
 &\leq \sum_{i \in \tau_1(n)}  2Q \delta^2/i^2 \leq 2  Q \beta \delta^2 , \label{eqn:chain}
\end{align}
since $D \geq 4/\Delta^2_{\min}$ and $\beta = \sum_{i=1}^\infty 1/i^2$. Hence, we have
\begin{align}
 \mathrm{P} ( \textrm{Perf}_{1}(n)  )   \geq   1 - 2  Q \beta \delta^2.   \notag
\end{align}
On event $\textrm{Perf}_{1}(n)$, it is always the case that the first teaching material that is shown is chosen according to the BF benchmark, independent of whether the eTutor explores or exploits the second slot. Hence given $\textrm{Perf}_{1}(n)$, the sample mean estimates that are related to $\textrm{Perf}_{2}(n)$ are always sampled from the distribution in which the first teaching material is shown according to the BF benchmark. Because of this, we have
\begin{align}
 \mathrm{P}  (\textrm{Perf}_{2}(n) |  \textrm{Perf}_{1}(n)   ) \geq   1 - 2  Q \beta \delta^2 .   \notag
\end{align}
Similarly, it can be shown that
\begin{align}
 \mathrm{P}  (\textrm{Perf}_{t}(n) |  \textrm{Perf}_{t-1}(n) , \ldots,   \textrm{Perf}_{1}(n)  )   \geq  1 - 2  Q \beta \delta^2.   \notag
\end{align}
Combining all of this and using (\ref{eqn:chain}) we get 
\begin{align}
\mathrm{P} (\textrm{Perf}(n) )  & \geq (1-  2  Q \beta \delta^2)^Q \notag \\
& \geq 1 - 2 Q^2 \beta \delta^2.   \notag \\
& = 1 - \epsilon ,
\end{align}
since $\delta = \sqrt{\epsilon}/(Q \sqrt{2 \beta})$. 

Next we bound $R_{e}(n)$. From the definition of $E_t(n)$, $t = 1, \ldots, Q$, we know that $|E_1(n)| \leq |{\cal X}| |{\cal A}| Q D \log (n/ \delta)$. Similarly, for $E_t(n)$, $t=2,\ldots, Q$, we have $|E_t(n)| \leq 2 |{\cal X}| |{\cal A}| Q D \log (n/ \delta)$. Hence, we have $|E(n)| \leq 2 |{\cal X}| |{\cal A}| Q^2 D \log (n/ \delta)$. Since the worst-case reward loss due to showing a suboptimal set of teaching materials to a student is at most $1$, we have
\begin{align}
R_{e}(n) \leq   2 |{\cal X}| |{\cal A}| Q^2 D \log (n/ \delta)/n .    \notag
\end{align}

Finally, the regret bound on $R(n)$ holds by taking the expectation.

\end{proof}

Theorem 1 implies that the average final exam score of students tutored by eTutor converges to the average final exam score of students tutored by BL (with probability at least $1-\epsilon$) which knows the expected final exam scores, and hence, how students learn for each sequence of teaching materials perfectly. Moreover, the regret gives the convergence rate, and since it decreases with $\log n/n$, eTutor converges very fast.

\section{Illustrative Results} \label{sec:simulations} 

We deployed our eTutor system for students who have already studied digital signal processing (DSP) one or more years ago, and the goal of this implementation of the eTutor is to have them refresh the material about $\emph{discrete Fourier transform}$ (DFT) in the minimum amount of time. 
Student contexts belong to ${\cal X} = \{0,1\}$, where for a student $i$, $x_i=0$ implies that she is not confident about her knowledge of DFT, and $x_i =1$ implies that she is confident about her knowledge of DFT.
${\cal Q}$ contains three ({\em remedial}) materials: one text that describes DFT and two questions that refreshes DFT knowledge. If a question is shown to the student and if the student's answer is incorrect, then the correct answer is shown along with an explanation. For each $q \in {\cal Q}$, we set the cost to be $c_q = 0.04 \times \theta_q$, where $\theta_q$ (in minutes) is the average time it takes for a student to complete material $q$.  The value of $\theta_q$ is estimated and updated based on the responses of the students.
Performance of the students after taking the remedial materials are tested by the same final exam. 

We compare the performance of eTutor with a {\em random rule} (RR) that randomly selects the materials to show and a {\em fixed rule} (FR) that shows all materials (text first, easy question second, hard question third). The average final score achieved by these algorithms for $n=100$ and  $n=500$ students are shown in Table \ref{table:numresults}. From this table we see that eTutor achieves 15,7\% and 1.1\% improvement in the average final score for $n=500$ compared to RR and FR, respectively. The improvement compared to FR is small because FR shows all the materials to every student. It is observed that the average final score of eTutor increases with $n$, which is expected since eTutor learns the best set of materials to show for each context as more students take the course.  In contrast, RR and FR are non-adaptive, hence their average final exam scores do not improve as more students take the course. 
For $n=500$, the average time spent by each student taking the course is $8.5$ minutes for eTutor which is 16.7\% and 50\% less than the average time it takes for the same set of students by RR and FR, respectively. eTutor achieves significant savings in time by showing the best materials to each student based on her context instead of showing everything to every student. 

\begin{table}
\centering
{\fontsize{10}{10}\selectfont
\setlength{\tabcolsep}{.1em}
\begin{tabular}{|l|c|c|}
\hline
 \# of students  &  $n=100$ & $n=500$\\
\hline
eTutor & (66.4, 8.7) & (75.8, 8.5)   \\
\hline
RR &  (62.4, 10.2) & (62.5, 10.2) \\
\hline
FR & (75.5, 17.0)  & (75.0, 17.0) \\
\hline
\end{tabular}
}
\caption{Comparison of eTutor with RR and FR: For each entry $(x, y)$, $x$ denotes the average final score (maximum = $100$) and $y$ denotes the time spent in minutes taking the course.}
\label{table:numresults}
\vspace{-0.2in}
\end{table}


 
\section{Conclusion}\label{sec:conc}

In this paper, we proposed a novel online education system called eTutor. 
While in this paper, eTutor was used to learn the best sequence of materials to show a specific student, eTutor can also be easily adapted to learn the best teaching methodology such as
what types of materials/examples to show (visual or not), what style of teaching to use etc.


\bibliographystyle{IEEEtran}
\bibliography{refs}

\begin{thebibliography}{10}
\providecommand{\url}[1]{#1}
\csname url@samestyle\endcsname
\providecommand{\newblock}{\relax}
\providecommand{\bibinfo}[2]{#2}
\providecommand{\BIBentrySTDinterwordspacing}{\spaceskip=0pt\relax}
\providecommand{\BIBentryALTinterwordstretchfactor}{4}
\providecommand{\BIBentryALTinterwordspacing}{\spaceskip=\fontdimen2\font plus
\BIBentryALTinterwordstretchfactor\fontdimen3\font minus
  \fontdimen4\font\relax}
\providecommand{\BIBforeignlanguage}[2]{{%
\expandafter\ifx\csname l@#1\endcsname\relax
\typeout{** WARNING: IEEEtran.bst: No hyphenation pattern has been}%
\typeout{** loaded for the language `#1'. Using the pattern for}%
\typeout{** the default language instead.}%
\else
\language=\csname l@#1\endcsname
\fi
#2}}
\providecommand{\BIBdecl}{\relax}
\BIBdecl

\bibitem{Survey}
C.~Brinton and M.~Chiang, ``Social learning networks: A brief survey,'' in
  \emph{Information Sciences and Systems (CISS), 2014 48th Annual Conference
  on}, March 2014, pp. 1--6.

\bibitem{Stanford}
A.~Anderson, D.~Huttenlocher, J.~Kleinberg, and J.~Leskovec, ``Engaging with
  massive online courses,'' in \emph{23rd International World Wide Web
  Conference}, 2014.

\bibitem{Coursera}
R.~F. Kizilcec, C.~Piech, and E.~Schneider, ``Deconstructing disengagement:
  analyzing learner subpopulations in massive open online courses,'' in
  \emph{Third Conference on Learning Analytics and Knowledge}, 2013, pp.
  170--179.

\bibitem{edX}
L.~Breslow, D.~E. Pritchard, J.~DeBoer, G.~S. Stump, A.~D. Ho, and D.~T.
  Seaton, ``Studying learning in the worldwide classroom research into edx's
  first mooc,'' \emph{Research and Practice in Assessment}, pp. 13--25, 2013.

\bibitem{AES1}
S.~Guven, ``Mltutor: a web-based educational adaptive hypertext system,'' in
  \emph{Lost in the Web - Navigation on the Internet (Ref. No. 1999/169), IEE
  Colloquium}, 1999, pp. 4/1--4/3.

\bibitem{AES2}
N.~Henze and W.~Nejdl, ``Adaptation in open corpus hypermedia,''
  \emph{International Journal of Artificial Intelligence in Education}, pp.
  325--350, 2001.

\bibitem{IT1}
A.~Mitrovic, ``An intelligent sql tutor on the web,'' \emph{International
  Journal of Artificial Intelligence in Education}, pp. 171--195, 2003.

\bibitem{IT2}
T.~Heift and D.~Nicholson, ``Web delivery of adaptive and interactive language
  tutoring,'' \emph{International Journal of Artificial Intelligence in
  Education}, pp. 310--324, 2001.

\bibitem{IT3}
K.~Forbes-Riley, D.~Litman, and M.~Rotaru, ``Responding to student uncertainty
  during computer tutoring: A preliminary evaluation,'' in \emph{Proceedings
  9th International Conference on Intelligent Tutoring Systems (ITS)}, 2008.

\bibitem{IT4}
S.~Schiaffino, P.~Garcia, and A.~Amandi, ``eteacher: Providing personalized
  assistance to e-learning students,'' \emph{Computers and Education}, vol.~51,
  no.~4, pp. 1744 -- 1754, 2008.

\bibitem{IT5}
A.~Zouhair, E.-M. En-Naimi, B.~Amami, H.~Boukachour, P.~Person, and
  C.~Bertelle, ``Intelligent tutoring systems founded of incremental dynamic
  case based reasoning and multi-agent systems (its-idcbr-mas),'' in
  \emph{Advanced Logistics and Transport (ICALT), 2013 International Conference
  on}, May 2013, pp. 341--346.

\bibitem{IT6}
S.~Piramuthu, ``Knowledge-based web-enabled agents and intelligent tutoring
  systems,'' \emph{Education, IEEE Transactions on}, vol.~48, no.~4, pp.
  750--756, Nov 2005.

\bibitem{IT7}
S.~Nafiseh and M.~Ali, ``Evaluation based on personalization using optimized
  firt and mas framework in engineering education in e-learning environment,''
  in \emph{E-Learning and E-Teaching (ICELET), 2013 Fourth International
  Conference on}, Feb 2013, pp. 117--120.

\end{thebibliography}

\end{document}